\documentclass[conference]{IEEEtran}
\IEEEoverridecommandlockouts
\usepackage{cite}
\usepackage{amsmath,amssymb,amsfonts}
\usepackage{hyperref}
\usepackage{graphicx}
\usepackage{textcomp}
\usepackage[T1]{fontenc}
\usepackage[utf8]{inputenc}
\usepackage{lmodern}

\usepackage[english]{babel}

\usepackage{algorithm}
\usepackage{algpseudocode}
\makeatletter
\renewcommand{\ALG@beginalgorithmic}{\small}
\makeatother

\usepackage{mathtools}
\usepackage{mathdots}
\usepackage{tikz-cd} 
\usepackage{wrapfig}
\usepackage{float}
\usepackage{quiver}
\usepackage{adjustbox}
\usepackage{stackengine,scalerel}
\newcommand\eye{\ensurestackMath{\stackinset{c}{}{c}{-.33pt}%
		{\bullet}{\bigcirc}}}

\makeatletter
\newcommand{\removelatexerror}{\let\@latex@error\@gobble}
\makeatother

\newcounter{cases}
\newcounter{subcases}[cases]

\DeclareMathOperator*{\argmin}{argmin}
\DeclareMathOperator*{\pre}{^\bullet}
\usepackage{amsthm}
\usepackage{tikz}
\usetikzlibrary{cd}
\usetikzlibrary{petri,positioning}
\newtheorem{theorem}{Theorem}
\newtheorem{lemma}[theorem]{Lemma}
\newtheorem{remark}{Remark}

\newtheorem{definition}{Definition}

\newtheorem{example}{Example}

\def\BibTeX{{\rm B\kern-.05em{\sc i\kern-.025em b}\kern-.08em
		T\kern-.1667em\lower.7ex\hbox{E}\kern-.125emX}}
\begin{document}
	\title{Timed Alignments\\
		\thanks{Neha Rino was funded by the International Master's Scholarships Program IDEX of Université Paris-Saclay.}
	}
	
	\author{\IEEEauthorblockN{Thomas Chatain}
		\IEEEauthorblockA{\textit{LMF, ENS Paris-Saclay,}\\
			\textit{CNRS, Université Paris-Saclay, Inria} \\
			Gif-sur-Yvette, France\\
			\url{thomas.chatain@ens-paris-saclay.fr}\\
			ORCID 0000-0002-1470-5074}
		
		\and
		\IEEEauthorblockN{Neha Rino}
		\IEEEauthorblockA{\textit{LMF, ENS Paris-Saclay,}\\
			\textit{CNRS, Université Paris-Saclay, Inria} \\
			Gif-sur-Yvette, France\\
			\url{neha.rino@ens-paris-saclay.fr}}
	}

	\maketitle
	
	\begin{abstract}
		The subject of this paper is to study conformance checking for timed models, that is, process models that consider both the sequence of events in a process as well as the timestamps at which each event is recorded. 
		Time-aware process mining is a growing subfield of research, and as tools that seek to discover timing related properties in processes develop, so does the need for conformance checking techniques that can tackle time constraints and provide insightful quality measures for time-aware process models. 
		In particular, one of the most useful conformance artefacts is the alignment, that is, finding the minimal changes necessary to correct a new observation to conform to a process model.  
		In this paper, we set our problem of timed alignment and solve two cases each corresponding to a different metric over time processes. For the first, we have an algorithm whose time complexity is linear both in the size of the observed trace and the process model, while for the second we have a quadratic time algorithm for linear process models. 
	\end{abstract}

	
	\section{Introduction}

	\subsection{Conformance Checking and Alignments} 
	
	Process mining studies vast systems through their event logs, and seeks to extract meaningful ways to model the underlying patterns or processes that govern the behaviour of the system in order to better understand the system, or predict future behaviour \cite{A16}.
	Once such a process model is obtained, through machine learning or other related techniques, it is natural to ask how one is sure the obtained model is a reasonable approximation of the system’s behaviour at all, especially given the lack of explainability in the blackbox approach ML takes to producing solutions.
	This is where conformance checking comes into the picture, as it is the art of judging the performance of a process model by relating modelled and observed behaviour of a process to each other, without depending on the origin of the model \cite{CDSW18}.
	Observed behaviour comes in the form of traces in an event log, as a sequence of events occurring during the functioning of the system, while process models are blueprints that describe what the underlying processes of any given system are supposed to look like.
	Measures of how well a model reflects system behaviour include fitness, precision, generalisation, and simplicity.
	We often do not want the system to always precisely generate any and all possible future system behaviour, but neither should they simply regurgitate the event log and accept no new system behaviours. 
	What is often much more useful is a process model that can, up to some small error factor, approximate any reasonable future system behaviour. 
	
	In Arya Adriansyah's seminal thesis \cite{Adriansyah2014AligningOA}, we obtain the notion of an alignment, that is, the minimal series of corrections needed to transform an observed event trace into the execution of the process model that most closely mimics it.  
	It is often seen as an execution of the synchronised product of the process model and a trace model i.e. a simple model that captures exactly one event log trace. 
	This gives us a series of edits, usually insertions or deletions, that can transform a process trace into the observed trace.  
	Alignments thereby help pinpoint exactly where inevitable deviations from expected behaviour occur, and the more distant the aligning word of a model is to its observed trace, the worse the model is at reflecting real system behaviour. 
	
	\subsection{Time-aware Processes}
	
	Process models can be represented using a variety of formal objects, including Petri nets. 
	They both offer a graphical means by which to represent concurrent systems, and a formal semantics for their execution, which allows one to mathematically analyse the same. 
	Assuming the event logs are a list of words over a finite alphabet (the set of possible discrete events), the problem of calculating the alignment has been extensively studied \cite{Adriansyah2014AligningOA} \cite{BCC21}. 
	The notion of distance used on these words used is usually either Hamming distance or Levenshtein's edit distance. 
	It is natural to want to study explicitly timed systems
	as by considering events along with their timestamps when mining processes, we can study the minimum delay between two events, or the maximum duration the system takes to converge upon a state, or check deadlines, all of which are highly relevant in real world applications \cite{Cheikhrouhou2014TheTP} \cite{Eder1999TimeCI} \cite{inbook}. 
	In addition, one may want to predict the timestamps of processes \cite{SSA11}. 
	In the process mining community, there are ways to use existing process model notation in order to denote some time constraints. BPMN 2.0 comes equipped with \emph{timer events} and can record absolute, relative, and cyclical time constraints. 
	For our purposes, we use an extension of Petri nets equipped with the ability to express the relative and cyclical time constraints that BPMN equipped with timers can. 
	What we are referring to are time Petri nets, that is, Petri nets augmented with the ability to record and check the duration it takes to fire a transition once enabled, using which one can impose certain constraints on the relationships between the timestamps of different events. 
	
	Time-aware process mining seeks to study both what sort of underlying processes govern system behaviour, and what sort of time constraints they can impose on when certain events can occur \cite{RMAW13} \cite{CRHA20}  \cite{AS21}. 
	The need for such a setting is clear, as specifying and checking concrete time constraints regarding the durations of events can ensure critical information about a system's behaviour. 
	As time-aware process mining grows popular, new quality measures and conformance checking techniques must be developed that are sensitive to temporal constraints, but so far in the study of alignments as a conformance checking artefact, we notice that the process model used is never time-aware. 
	For this, crucially, one needs to define distance functions that can meaningfully compare and separate different time processes. 
	This paper seeks to provide a framework by which to do the very same, and set and solve the alignment problem for time-aware processes.
	
	\subsection{Three Distances, and Algorithms to Align them}
	
	In this paper, we propose three different distance functions over timed words, and study the alignment problem for the first two. 
	The first is essentially Manhattan distance, and hence very standard, and has a quadratic time algorithm in the case of a structurally restricted class of models (linear causal processes).
	For all other types of models, we present an encoding in simplex that solves the problem, albeit with exponential worst-time complexity.
	The second distance function utilises the structure of the process model to study distances between time processes, and we present a straightforward algorithm for solving the alignment in this case, whose time complexity is linear in the size of the causal process and the transition set of the model. The second setting, and the new notion of delay edits that it represents, provide insight into time processes and their relevant constraints. 
	
	\section{Preliminaries}
	We represent events as pairs $(a,t)$ where $a \in \Sigma$ is the name of the action, and $t$ denotes the time at which said action was taken.
	
	\begin{definition}
		A timed trace is a sequence $\gamma \in (\Sigma \times \mathbb{R}^+)^*$ of timed events, seen as a timed word. 
	\end{definition}
	
	We will often ignore the untimed parts of timed words, i.e., the projection onto $\Sigma^*$, leaving just the timestamps, a sequence belonging to $\mathbb{R}^{+*}$. 
	
	The timed process model we use here is a labelled time Petri net. 
	
	\begin{definition}[Labelled Time Petri Net] \label{tpn} A labelled \textit{time Petri net} (or TPN) is a tuple $N = (P, T, F, SI, \Sigma, \lambda, M_0, M_f)$, where $P, T$ are disjoint sets of places and transitions respectively,  $F \subseteq (P \times T) \cup (T \times P)$ is the flow relation,  $SI : T \rightarrow \mathbb{I}$ is the static interval function, where $SI(t) = (Eft(t), Lft(t))$ such that $Eft$ stands for earliest firing time, and $Lft$ for latest firing time, $\lambda : T \rightarrow \Sigma$ is the labelling function, labelling transitions with actions from the action set $\Sigma$, and  $ M_0 , M_f: P \to \mathbb{N}$ are the initial and final markings. 
	\end{definition}
	
	Given a transition $t \in T$ we define the pre-set of $t$ as $\pre t = \{p \in P | (p, t) \in F\}$ and its post-set similarly is defined as $t \pre = \{p \in P | (t, p) \in F\}$, (the presets and post-sets of places are defined similarly). A transition $t$ of a time Petri net is \textit{enabled} at marking $M$ iff $\forall p \in \pre t :  M(p) > 0$. 
	The set of all enabled transitions at a marking $M$ is denoted by \textit{Enabled(M)}. 
	
	A \textit{state} of a TPN $N = (P, T, F, SI, \Sigma, \lambda, M_0, M_f)$ is a pair $S = (M, I)$, where $M$ is a marking of $N$ and $I : Enabled(M) \rightarrow \mathbb{R^+}$ is called the \textit{clock function}. The initial state is $(M_0, \mathbf{0})$, where $\mathbf{0} $ is the zero function. 
	
	A transition $t$ is fireable from state $S = (M, I)$ after delay $\theta \in \mathbb{T}$ iff $t$ is enabled at $M$, and updating the clock function to increment by delay $\theta$ will keep $t$'s new clock value $I(t)  + \theta$ in the range $[Eft(t), Lft(t)]$ as determined by the static interval function. 
	
	Once a fireable transition $t$ is fired, the marking and clock function are both updated to reflect the firing, as defined below : 
	
	\begin{definition}[Firing Rule]
		When a transition $t$ fires after time $\theta$ from state $S = (M, I)$, the new state $S' = (M'', I')$ is given as follows : $$M'' = M' \cup t \pre \textrm{ where } M' = M \setminus \pre t$$ $$I'(t) = \begin{cases}
			I(t) + \theta & \textrm{If } t \in Enabled(M')\\
			0 & \textrm{If } t \in Enabled(M'') \setminus Enabled(M')\\
			\textrm{Undefined} & \textrm{Otherwise}\\
		\end{cases}$$
	\end{definition}
	This is also denoted by $(M, I)[t\rangle(M'', I')$
	
	An important feature of time Petri nets is the notion of urgency, that is, if $t$ is enabled at marking $M$ and has clock value $Lft(t)$, it must fire, or another transition must fire at the same instant disabling $t$.  
	
	A valid execution of the model begins at the initial marking, fires a sequence of transitions (representing a series of activities occurring at certain times) and at the end of the firing sequence reaches $M_f$, with any clock function $I$. 
	
	\begin{definition}[Language of a time Petri net]
		A word $w = (a_0, a_1, \dots a_n) \in \Sigma^*$ is in the language of the labelled time Petri net $\mathcal{L}(N)$ if there is a fireable sequence of transitions $(t_0, t_1 \dots t_n) \in T^*$ such that $\lambda((t_0, t_1, \dots, t_n)) = w $ and they transform the initial marking into the final one, that is, for some clock function $I$ on $M_f$, $$(M_0, \mathbf{0}) [t_0, t_1, \dots t_n\rangle (M_f, I)$$
	\end{definition}
	
	\begin{example}\label{TPN}
		Consider the following example of a time Petri net $N$: \\
		\adjustbox{scale = 0.8}{
			\begin{tikzcd}
				{} & {} & \bigcirc & {\fbox{b}} & \bigcirc && {} \\
				\eye & {\fbox{a}} && {\fbox{c}} && {\fbox{f}} & \bigcirc \\
				{}&\fbox{d}& \bigcirc & {\fbox{e}} & \bigcirc \\
				\arrow[curve={height=-6pt}, from=2-2, to=1-3]
				\arrow[curve={height=6pt}, from=2-2, to=3-3]
				\arrow[from=2-1, to=2-2]
				\arrow[ from=3-3, to=3-2]
				\arrow[from=3-3, to=3-4]
				\arrow[from=3-4, to=3-5]
				\arrow[curve={height=12pt}, from=1-3, to=2-4]
				\arrow[from=1-3, to=1-4]
				\arrow[from=1-4, to=1-5]
				\arrow[curve={height=12pt}, from=2-4, to=1-5]
				\arrow[curve={height=-6pt}, from=1-5, to=2-6]
				\arrow[curve={height=6pt}, from=3-5, to=2-6]
				\arrow[from=2-6, to=2-7]
				\arrow["{[1,1]}"{description, pos=0.1}, draw=white, from=1-4, to=2-4]
				\arrow["{[0,2]}"{description, pos=0.1}, draw=white, from=2-4, to=1-4]
				\arrow["{[1,3]}"{description, pos=0.2}, draw=white, from=3-2, to=3-1]
				\arrow["{[1,4]}"{description, pos=0.3}, draw=white, from=3-4, to=2-4]
				\arrow["{[0, 3]}"{description, pos=0.2}, draw=white, from=2-6, to=1-7]
				\arrow["{[0, \infty)}"{description, pos=0.2}, "\lrcorner"{text=white, anchor=center, pos=0.125, rotate=180}, draw=none, from=2-2, to=1-1]
				\arrow[shift right=2, curve={height=6pt}, from=3-2, to=3-3]
			\end{tikzcd}
		}\\
		
		One possible execution of $N$ would be for the firing sequence $$w = (a, 1)(b, 2)(d, 3)(e, 4)(f,5)$$ The initial marking only has $a$ enabled, and firing $a$ at time $1$ updates the marking by removing the token from $a$'s preplace and filling its two post-places, thereby enabling $b, c, d, e$. Now, note by urgency that $c$ can never fire past time $2$ as then $b$ would have already fired, but in this execution this doesn't matter as $b$ fires at $2$ regardless, filling one preplace of $f$. Then $d$ fires, having been enabled for $2$ units of time, thereby resetting the enabling time of both $d$ and $e$ to zero, so that when $e$ does fire it has only been most recently enabled for $1$ unit of time. Now, finally, $f$ is enabled, and fires after $1$ unit of having been enabled. 
	\end{example}
	
	Now we will want to build some vocabulary to help us talk about individual executions of timed words on time Petri nets. 
	
	\begin{definition}[Causal Net] \label{cn}
		A \textit{causal net} $CN = (B, E, G)$ is a finitary, acyclic net where 
		$$\forall b \in B : |b \pre| \leq 1 \wedge | \pre b| \leq 1.$$
	\end{definition}
	
	This can be viewed as the original Petri net itself, but every time a place is revisited, it is copied afresh to ensure that the execution only ever moves forward. 
	Hence, multiple elements of $B$ or $E$ map to the same element in $P$ or $T$ respectively, a notion that is expressed via the map $p$ defined below. 
	
	\begin{definition}[Homomorphism] \label{hom}  Let $N$ be a time Petri net with place set $P$ and transition set $T$, and ${CN = (B, E, G)}$ be a causal net. 
		A mapping $p : B \cup E \rightarrow P \cup T$ is a \textit{homomorphism} if $p(B) \subseteq P$,  $p(E) \subseteq T$ and $\forall e \in E$ the restriction of $p$ to $\pre e$ is a bijection between $\pre e$ and $\pre p(e)$ and the restriction of $p$ to $e \pre $ is a bijection between $e \pre $ and $p(e) \pre $ and the restriction of $p$ to $Min(CN)$ is a bijection between $Min(CN)$ and $M_0$. 
	\end{definition}
	
	\begin{definition}[Causal Process]\label{cp} A \textit{causal process} of a TPN is a pair $(CN, p)$ where $CN$ is a causal net and $p$ is a homomorphism from $CN$ to $TPN$. 
	\end{definition}
	
	Using $p$, elements of $CN$ are identified with their corresponding \emph{originals} in the time Petri net, and as a result, any causal process corresponds uniquely to an untimed run on the untimed version of a given time Petri net. 
	
	\begin{definition}[Timing Function]\label{tau} A \textit{timing function} $\tau : E \rightarrow \mathbb{R}$ is a function from events of a causal process into time values. 
	\end{definition}
	
	\section{The Alignment Problem in Timed Settings}
	The alignment problem, given a trace from an event log and a process model, involves finding a valid execution of the process model that is closest to the trace under some metric.
	
	\begin{definition}[The General Alignment Problem]
		Given a process model $N$ denoted by a Time Petri Net and a timed trace $\sigma$ we wish to find a timed word $\gamma \in \mathcal{L}(N)$ such that $d(\sigma, \gamma) = \min_{x \in \mathcal{L}(N)}{d(\sigma, x)}$ for some distance function $d$ on timed words. 
	\end{definition}
	
	In the untimed setting, this is viewed as a problem of minimizing cost over a series of edit moves, either insertions (a model move) or deletions (a trace move). When aligning timed words, clearly there are two crucial aspects to the problem, which are, how one deals with deviations in the action labels themselves (the action labels), versus how one deals with deviations in just the timing properties of the word (the timestamps). The alignment problem has been extensively studied for the untimed case \cite{Adriansyah2014AligningOA} \cite{BCC21}, 
	but the timed setting is more complex. 
	
	\begin{example}
		Consider the process model of Example \ref{TPN} and an observed trace $\sigma_1 = (a, 0)(a, 1)(b,2)(d,3)(e,3)(f,5)$. Clearly, there is an extra $a$ event that cannot have occurred as no process trace has more than one $a$, and also, the $e$ event fires too early, it should have waited at least one time unit after the firing of event $d$. 
		
		On the other hand, consider an observed trace of the form $\sigma_2 = (a, 1)(b,1)(d,3)(e,4)(f,5)$. Here, either we can assume the $b$ fired too early, and correct to $(b,2)$, or that was supposed to be a $c$ event as those two are parallel choices. 
		
		Thirdly, consider an observed trace of the form $\sigma_3 = (a, 1)(d,1)(d,2)(e,4)(f,5)$. Clearly one of the $d$'s needs to be converted to a $b$ or a $c$ event, but which $d$? The reasonable choices seem to be either correcting $(d,1)$ to $(c,1)$ or $(d,2)$ to $(b,2)$ or $(c,2)$. 
		Timestamps are attached to the letters they belong to, so in a sense, firing event $b$ at time $2$ as compared to firing event $c$ at time $1$ are not easily comparable, as the time constraints on different events are often different. In this manner, it is not clear how meaningful it would be to define edit moves that transform timestamps that are attached to different tasks. 
	\end{example}

	This is why, we first restrict our attention to the case where the untimed part of $\sigma$ does match the model, that is to say, if $\sigma = (w, \tau)$ where $w$ is the action labels and $\tau$ is the timestamp sequence. then we assume that $w$ has a valid causal process $(CN, p)$ over $N$. Clearly this degenerate case of the problem needs to be solved if the general timed alignment problem is to be solved, and we argue that this case is interesting and complex in its own right. 
	
	This ensures that we can compare the timing constraints on analogous parts of the process, and hence reason meaningfully about deviances from the model that are purely based on the timing of the process. This gives us the following problem : 
	
	\begin{definition}[The Purely Timed Alignment Problem]
		Given a process model $N$ denoted by a Time Petri Net and a timed trace $(w, \sigma)$, and a valid causal process of the untimed part $w$, we wish to find a valid timing function $\gamma$ such that $d(\sigma, \gamma) =\min_{x \in \mathcal{L}(N)}{d(\sigma, x)}$ for some distance function $d$ on timing functions. 
	\end{definition}
	
	Hence, for the first few sections of this paper, we will ignore the untimed parts of the words, i.e we assume the word $\sigma$ we wish to align comes with a valid causal process $(CN, p)$ of the model, and seek to find timing functions on said causal process that are both valid, and minimize distance to $\sigma$ under the metric of choice. 
	
	We shall revisit the general timed alignment problem in section \ref{align}, and show how our methods 
	can be adapted using existing techniques to provide approaches for solving the general 
	problem. 
	
	\subsection{Edit Moves on Timing Functions}\label{dualrep}
	Now we come to the problem of deciding how to compare two timing functions over the same causal process, and quantify how close they are to each other. 
	
	Much like Levenshtein's edit distance, popularly used in the untimed case of the alignment problem, we view the definition of these distances as an exercise in cost minimisation over the set of all transformations between two words. 
	In order to formalise the same, we need to define what the valid moves of such a transformation could be. 
	We define \textit{moves} as functions that map one timing function over $(CN, p)$ to another. 
	What sort of functions are useful notions of transformation on a timed system? 
	
	\begin{example}
		Let us go back to example \ref{TPN}, and study the process model $N$ presented there. 
		
		Now, say we had the following words that did not fit the model, and we wished to analyse how best to modify them to fit them back into the model. 
		
		We start with $(u_1, \sigma_1) = (abdef)(1, 3, 4, 6, 6)$. 
		One candidate for the closest valid execution would be if the timing function $\sigma_1$ were replaced by $\gamma = (1, 2, 4, 6, 6)$, and it feels reasonable to say that the cost for aligning $\sigma_1$ to this $\gamma$ is 1. 
		A way to arrive at this conclusion is by trying to execute $\sigma$'s firing sequence, and noticing that only the guard for $b$ fails. If $b$ were shifted to fire at 2 instead, the whole run would execute without a hitch. This sort of local, almost typographical error can often happen in systems, and it is the simplest kind to fix.  
		
		We start with $(u_1, \sigma_2) = (abdef)(1, 2, 5, 9, 12)$. 
		The closest valid timing function would try to preserve the positions of $a$ and $b$ but the moment it tries to fire $d$, $\sigma_2$ runs late throughout. The closest we can get, intuitively, is to fire every transition in the lower branch as late as possible, giving $\gamma = (1, 2, 4, 8, 11)$. 
		Now, when trying to compare these firing sequences, we can view it just like we did for $\sigma_1$, as $d$, $e$ and $f$ all fire later than they should, each timestamp is moved back once, giving an aligning cost of 3. 
		There is however, another way to see this deviation. 
		This cascading chain of errors can be fixed if $d$ is moved back to fire at 4, and all the relative relationships between $d$ and its successors are preserved. 
		This views the tasks $e$ and $f$ as only caring about when $d$ ended, which makes sense, because they depend on $d$'s completion. 
		This means, the switch from the timestamp series $(5,9,12)$ to $(4, 8, 11)$ can be viewed as only a cost 1 edit
		. 
		This is a slightly more complex error to conceive of, but it reflects the fact that if a delay at the beginning caused the whole process to run late, the important thing to fix is just that initial offset, and the rest of the process will now conform to the model as needed. 
	\end{example}
	
	Based on the above example, we naturally arrive at two types of moves. 
	
	We define a \textit{stamp move} as a move that translates the timing function only at a point, i.e., that edits a particular element of the timestamp series $\tau$. 
	
	\begin{definition}[Stamp Move]
		Given a timing function $\gamma : E \to \mathbb{R}$, formally, we define this as : 
		
		$\forall x \in \mathbb{R}, e \in E : stamp(x, e)(\gamma) =\gamma' $ where
		$$\forall e' \in E : \gamma'(e')  = \begin{cases}
			\gamma(e') + x & e' = e \\
			\gamma(e')& otherwise
		\end{cases}$$
		
	\end{definition}
	
	The next type of move we describe is the more novel and interesting \textit{delay move}. 
	Here, we sought to leverage the structure of the process model itself, by reflecting the causal relationships the pre-order $G$ causes. 
	If an event is $G$-reachable from another, that means the first event is strictly in the causal history of the second so changing the timestamp of such a causal predecessor has consequences for all of its causal descendents, while leaving any causally unrelated (i.e, non-$G$-reachable) events undisturbed. 
	
	In other words, a delay move at $e$ will preserve relative relationships between timestamps in the future, at the cost of shifting the timestamp of every causal descendent of $e$ by the same amount. 
	\begin{definition}[Delay Move]
		Given a timing function $\gamma : E \to \mathbb{R}^n$, we define a delay move applied to it as follows : 
		
		$\forall x \in \mathbb{R}, e \in E : delay(x, e)(\gamma) = \gamma'$ where 
		$$\forall e' \in E : \gamma'(e') = \begin{cases}
			\gamma(e') + x & e' \geq_G e \\
			\gamma(e')& otherwise
		\end{cases}$$
		
	\end{definition}
	
	Now, armed with these types of moves, three natural notions of distance can be constructed. 
	
	\begin{definition}[Stamp Only Distance : $d_t$]
		Given any two timing functions $\tau_1, \tau_2$ over the same causal process $(CN, p)$, we define the stamp-only distance $d_t$ as follows : 
		$$d_t(\tau_1, \tau_2) = \min \{cost(m) | {m \in Stamp^*}, {m(\tau_1) = \tau_2}\}$$
	\end{definition}
	
	\begin{definition}[Delay Only Distance : $d_\theta$]
		Given any two timing functions $\tau_1, \tau_2$ over the same causal process $(CN, p)$, we define the delay-only distance $d_\theta$ as follows : 
		$$d_\theta(\tau_1, \tau_2) = \min \{cost(m) | {m \in Delay^*}, {m(\tau_1) = \tau_2}\}$$
	\end{definition}
	
	\begin{definition}[Mixed Moves Distance : $d_N$]
		And thirdly, given any two timing functions $\tau_1, \tau_2$ over the same causal process $(CN, p)$, we define the mixed move distance $d_N(\tau_1, \tau_2)$ as follows : 
		$$\min \{cost(m) | {m \in (Stamp\cup Delay)^*}, {m(\tau_1) = \tau_2}\}$$
	\end{definition} 
	
	\begin{example}\label{stamptoy}
		Consider the following example : 
		$$\eye \longrightarrow \underset{[0,1]}{\square} \longrightarrow \bigcirc \longrightarrow \underset{[2,2]}{\square} \longrightarrow \bigcirc \longrightarrow \underset{[1,1]}{\square} \longrightarrow \bigcirc$$
		Now for this $N$, let the observed trace ${\sigma = (3, 4, 5) \not \in \mathcal{L}(N)}$. 
		
		The best $d_t$ alignment for the example in the diagram below is $\gamma = (1,3,4)$ with minimum cost $d_t(\sigma, \gamma) = 4$. 
		
		The best $d_t=\theta$ alignment for the example in the diagram below is also $\gamma = (1,3,4)$, but this time with minimum cost $d_\theta(\sigma, \gamma) = 3$, evidenced by the move sequence $(delay(-2, 1)delay(+1, 2))$. 
		
		And lastly, the best $d_N$ alignment for the example in the diagram below is also $\gamma = (1,3,4)$, the sequence of moves being one stamp and one delay move at the start, $m = stamp(-1, 1)delay(-1, 1)$, and now with minimum cost $d_N(\sigma, \gamma) = 2 < \min \{d_t(\sigma, \gamma), d_\theta(\sigma, \gamma)\}$. 
	\end{example}

	\section{Results and Algorithm : Stamp Only Setting}
	\begin{lemma}[Stamp only distance : $d_t$]\label{manhattan}
		$d_t$ as defined above is equivalent to the \textit{Manhattan distance} or \textit{taxicab distance} between two points of $\mathbb{R}^n$, where $|E| = n$, say $\mathbf{t} = (t_1, t_2, \dots t_n) $ and $\mathbf{s} = (s_1, s_2, \dots s_n)$, defined as $$d_t(\mathbf{t}, \mathbf{s}) = \sum_{i = 1}^{n} |t_i - s_i|$$
	\end{lemma}

	This lemma also allows us to restate the alignment problem for $d_t$, as we see below. 
	
	\subsection{Casting Alignment as a Linear Programming Problem}
	The alignment problem in the stamp-only setting can be viewed as a problem of convex optimization, minimizing the stamp only cost from a fixed observed trace $\sigma$, which by Lemma \ref{manhattan} is $cost(\tau|_i) = \sum_{j = 1}^{i} |\tau_j - \sigma_j|$ over the set of valid timestamp series $\tau$ for the model $N$. 
	What remains is to show that the space of valid runs of the model is a convex set, and this set has been shown to be convex in \cite{BM82}. 
	
	Hence, we claim that using an efficient encoding of the target function using the firing domain formulation described above, the problem can be solved using the simplex algorithm, which is used to solve convex optimization problems in good time in practice, and having worst case exponential time complexity. 
	
	\begin{theorem}\label{convexset}
		Given a bounded time Petri Net $N$, an observed trace $\sigma$, and its causal process $CN = (B, E, G)$, the alignment problem can be viewed as seeking the vector $(\gamma_1, \gamma_2, \dots \gamma_{|E|})$ that minimizes the quantity $\sum_{e \in E} |\gamma_e - \sigma_e|$ over the set $\gamma \in \mathcal{L}(N)$. 
		
		This can be cast as a linear programming problem, i.e, $$\underset{\alpha - \beta \in \mathcal{C}, \alpha, \beta \geq 0}{\mathsf{minimize}} \hspace{0.5em} \alpha + \beta  $$
		
		Where $\mathcal{C}$ is a convex set. 
	\end{theorem}
	
	\begin{remark}
		By the above theorem we see that the alignment problem for $d_t$ is easily cast as a simplex instance, and can be solved by any linear programming solver. 
		Hence it can be solved in polynomial average running time, but this approach does have theoretically exponential worst time complexity. 
		This bound is not necessarily tight, but either improving it or proving its tightness are left for future work. While the general case seems to be hard to solve efficiently, there is a subclass of process models for which the problem is significantly more tractable, and moreover, has a quadratic time solution. 
	\end{remark}
	
	\subsection{Summing Minimum Cost Graphs} 
	
	We are given a process model $N$, an observed trace $\sigma$, and its underlying linear causal process $CN = (B, E, G)$ with homomorphism $p$ for the untimed run of $\sigma$ on $N$. 
	
	Suppose we tried to express the cost of aligning a prefix of $\sigma$ to some timing function $\tau$ on the truncation of the linear causal process up to the $i$th transition, call it $cost$. 
	$$cost(\tau|_i) = \sum_{j = 1}^{i} |\tau_j - \sigma_j|$$
	
	Now, note the following equation holds for the $(i+1)$th static interval constraint $[a, b]$ : 
	$$\min_{\tau|_{i+1} \in \mathcal{L}^{i+1}(N)}cost(\tau|_{i+1}) = $$$$\min_{\tau_{i+1} - \tau_i\in [a, b]} (|\tau_{i+1} - \sigma_{i+1}| + \min_{\tau|_i \in \mathcal{L}^i(N) }cost(\tau|_i))$$
	
	This suggests that the minimal cost of aligning prefixes $\sigma$ to prefixes of the model can be recursively calculated as a function of the value of the last timestamp in the alignment of a given prefix. 
	$$g_i(d) = \min_{\tau|_i \in \mathcal{L}^i(N) \wedge \tau_i = d} cost(\tau|_i)$$
	
	Now, the previous equation can be expressed in terms of the $g$ functions as follows : 
	$$g_{i+1}(d) = |d - \sigma_{i+1}| + \min_{d-d' \in [a,b]} g_i(d')$$
	
	If we can calculate the family of functions $\forall i \leq n : g_i$ efficiently, then the minimal cost for aligning $\sigma$ to $N$ is simply $\min_{d} g_n(d)$

	\begin{example}\label{stampcomp}
		
		We take a moment to look at an example of the computation of this family of functions we proposed by considering the following $d_t$-alignment problem : 
		
		We first have the process model $N$ : 
		\\
		\adjustbox{scale = 0.7}{
			\begin{tikzpicture}
				
				\node[place, tokens = 1] (p1) at (0,2) {}; 
				
				\node[transition, label = {below : $[0,1]$}] (t1) at (1.5,2) {}; 
				
				\node[place] (p2) at (3,2) {}; 
				
				\node[transition, label ={ below : $[2,2]$}] (t2) at (4.5,2) {}; 
				
				\node[place] (p3) at (6,2) {}; 
				
				\node[transition, label ={ below : $[2,4]$}] (t3) at (7.5,2) {}; 
				
				\node[place] (p4) at (8.5,1) {}; 
				
				\node[transition, label ={ $[0,1]$}] (t4) at (7.5,0) {}; 
				
				\node[place] (p5) at (6, 0) {}; 
				
				\node[transition, label ={ $[0,0]$}] (t5) at (4.5, 0) {}; 
				
				\node[place] (p6) at (3, 0) {}; 
				
				\node[transition, label ={ $[2,4]$}] (t6) at (1.5, 0) {}; 
				
				\node[place] (p7) at (0,0) {}; 
				
				\draw[thick] (p1) edge[post] (t1)
				(t1) edge[post] (p2)
				(p2) edge[post] (t2)
				(t2) edge[post] (p3)
				(p3) edge[post] (t3)
				(t3) edge[post] (p4)
				(p4) edge[post] (t4)
				(t4) edge[post] (p5)
				(p5) edge[post] (t5)
				(t5) edge[post] (p6)
				(p6) edge[post] (t6)
				(t6) edge[post] (p7);
				
			\end{tikzpicture}
		}\\
		The observed trace is $\sigma = (1,2,4,6,6,8)$. 

		We begin to compute the family of functions $g_i$, starting with $g_1(x) = |1-x|$ in the domain $[0,1]$.
		Utilising the recursion, we proceed to compute $g_2$.  
		\begin{figure}[!t]
			\centering 
			\includegraphics[width = 0.5\textwidth]{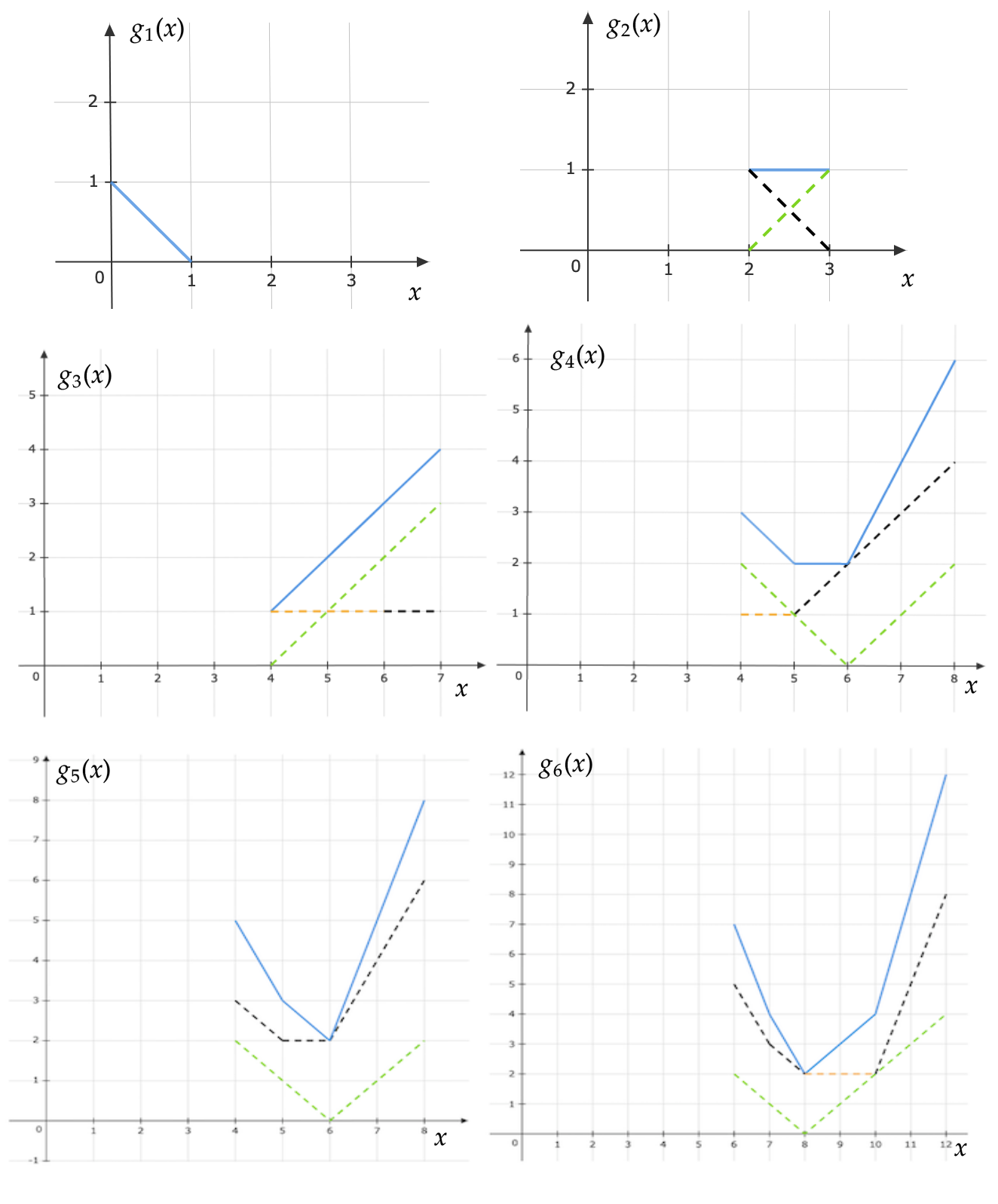}
			\caption{Visual representation of a run of the algorithm on Example \ref{stampcomp}}
		\end{figure}
		The new domain is $[2,3]$, a translation of the previous one because the guard on the second transition is the point $[2,2]$. 
		Now, we want the sum of both $|x-2|$ and $\min_{x' - x \in [2,2]}g_1(x')$. 
		The first of these functions, i.e., the contribution of the last timestamp to the error, is pictured in figure 1 in green dashed lines. 
		The second component is pictured in black dashed lines, and is a translation of $g_1$, as $\min_{x' - x \in [2,2]}g_1(x') = g_1 - 2$. 
		
		We construct $g_3$ similarly, by summing the contribution of the last timestamp, $|x - 4|$, to  $\min_{x' - x \in [2,4]}g_2(x')$. 
		We now find an orange segment in the second component :  This is used to demarcate any additional segments in  $\min_{x' - x \in [a_i, b_i]}g_{i-1}(x')$ not present in $g_{i-1}$ (if any). 
		We see the effect of the $\min$ operation is to attach a segment with the same length as the $i$th static interval constraint, at the lowest point of $g_{i-1}$. 
		Now we begin to see a simple way to construct these graphs visually : The next black graph is a translation of the previous blue one, with an orange line segment or the same length as the next static interval constraint inserted at the lowest point, and the next green graph is $|x-\sigma_i|$ in the new domain, and summing these two gives the next blue graph. 
		Even summation is easier than it looks, as these are all piecewise linear graphs, so one can simply increment the slopes of all the segments of the black-orange curve by one if they're to the right of $\sigma_i$, and by minus one if they're to the left, and translate the height up as needed. 
		Clearly the orange, black and green segments of the next graph are obtained like this, and $g_4$ is their sum. 
		
		The procedure for constructing $g_5$ and $g_6$ are identical, and now finally we see that the minimum cost of $2$ is achieved in $g_6$ by setting the last timestamp to $8$. 
		Now, we know the contribution of the $8$ to the error is zero, so this means we can deduce the corresponding value in $g_5$ must have been $2$ as well, achieved when the fifth timestamp was set to $6$. 
		Now $6$ gives a contribution of $0$ again to the error, letting us look at the preimage of $2$ once again on $g_4$, and so on. 
		We can backtrack in this manner to build an optimal alignment, such as in this case, $(0,2,5,6,6,8)$. 
	\end{example}

	With this context, we present the following algorithm that progressively calculates the graphs of $g_i$ (represented as lists of segments, each characterised by their slopes and x-projections of endpoints) each round, and claim it is both correct and efficient. 
	The algorithm precisely executes the above procedure, where the subfunction \textsc{graphMin} takes $g_i(d)$ and outputs $$\min_{d-d' \in [a,b]} g(d') = \begin{cases}
		g(d-a) & l + a \leq d < m + a \\
		g(m) & m+a \leq d \leq m' + b \\
		g(d+b) & m' + b < d \leq h + b   \\
	\end{cases}$$ Which is exactly the process of taking the last blue graph and creating the next orange-black graph, and subfunction \textsc{graphAddMod} takes $\min_{d-d' \in [a,b]} g(d')$ and adds $|d - \sigma_{i+1}|$ to it, which is precisely adding the green component.  Lastly, the \textsc{backTrack} algorithm backtracks through the list of graphs $\{g_i\}_n$ constructed to reverse engineer a word that aligns with the minimal cost calculated. A more detailed algorithm for backtracking is included in the appendix. 
	
	\begin{figure}[!t]
		\removelatexerror
		\begin{algorithm}[H]
			\caption{$d_t$ Algorithm}
			\label{StampAlgo}
			\begin{algorithmic}
				\State \textbf{Input : }  $\sigma$, a list of static interval constraints for the linear model $SI$
				\State \textbf{Output : }  $cost = \min_{x \in \mathcal{L}(N)}{d_t(x, \sigma)}$, $\gamma \in \mathcal{L}(N)$ such that $d_t(\gamma, \sigma) = cost$
				\State \textbf{Object :} $graph = \{(left, slope, right) | graph[i].left = graph[i-1].right\}$
				\Procedure{StampOnlyAlgo}{$\sigma, SI$}
				\State Initialise System 
				\State $[a, b] \gets SI[i]$
				\State $graph \gets \{(a, 0, b)\}$
				\State \textsc{graphAddMod}($graph, \sigma[0]$)
				\State $graphlist \gets \{graph\}$
				\State $i \gets 1$
				\State $cost \gets |a - \sigma[0]|$
				
				\While{$i < |sigma|$}
				\State $[a, b] \gets SI[i]$
				\State \textsc{graphMin}($graph, a, b$)
				\State \textsc{graphAddMod}($graph, \sigma[i]$)
				\State $graphlist.append(graph)$
				\State $cost \gets cost + |a - \sigma[i]|$
				\EndWhile
				
				\State $i \gets 0$ 
				\State $s \gets graph[0].slope$
				
				\While{$graph[i].slope < 0$}
				\State $s \gets graph[i].slope$ 
				\State $cost \gets cost + s * (graph[i].right - graph[i].left)$
				\State $i \gets i + 1$
				\EndWhile
				
				\State $\gamma \gets \textsc{backTrack}(\sigma, graphlist, cost, SI)$
				\State \textbf{Return : } $cost, \gamma$
				\EndProcedure
			\end{algorithmic}
		\end{algorithm}
	\end{figure}
	
	\begin{theorem}\label{StampAlgoCorr} Algorithm \ref{StampAlgo} is correct, i.e., given a linear causal process of a time Petri Net $N$, its result $\gamma$ has the properties :	
		\begin{enumerate}	\item $\gamma \in \mathcal{L}(N)$
			\item $\forall x \in \mathcal{L}(N)  : d_t(\gamma, \sigma)  \leq  {d_t(x, \sigma)}$
		\end{enumerate}
	\end{theorem}
	
	\begin{remark} Note that the subfunctions \textsc{graphMin} and  \textsc{graphAddMod} both run through the list of segments of the graph once each, and hence are linear in the sizes of their inputs, and Algorithm \ref{StampAlgo} calls each of these subfunctions once for each letter of the observed trace, each time on a graph with size linear in the current prefix.  
		So the cost calculation section of Algorithm \ref{StampAlgo} has quadratic time complexity in the size of the input. 
		
		The subfunction \textsc{backTrack} runs a binary search on each stored graph in $graphlist$ to find a trace in the language that does indeed give the minimum cost, letter by letter, and so overall takes $O(nlogn)$ time, keeping the overall time complexity  $O(n^2)$ where $n = |\sigma|$.
	\end{remark}
	
	\section{Results and Algorithm : Delays Only Setting}
	With delay edits, which constitute a new way of thinking transformations over time-series, comes a new perspective with which we can view timing functions over causal processes. 
	There is of course the standard definition, $\tau : E \to \mathbb{R}^+$ that assigns to each event a timestamp that records exactly when the event occurs. 
	
	Instead, thinking along the lines of durations between events occurring, we will often benefit from considering the following representation when speaking about delay moves, defined as to view a timed word not in terms of its absolute timestamps, but by the delays between them. 
	
	\begin{definition}[Flow Function]
		Given a causal process $(CN, p)$ over $N$, where $CN = (B, E, G)$, and a (not necessarily valid) timing function $\tau : E \to \mathbb{R}^+$, we first define the \textit{flow} function of $\tau$, $f_{\tau} : E \to \mathbb{R}^+$ such that 
		$$f_{\tau}(e) = \begin{cases} 
			\tau(e) & \pre \pre e = \emptyset\\
			\tau(e) - \tau(e') &e' \in   \pre \pre e,\\ & \tau(e') = \max\limits_{e'' \in \pre \pre e} {\{\tau(e'')\} \cup \{0\}}\\
		\end{cases}
		$$
		
	\end{definition}
	
	\begin{example}\label{wordandflow}
		
		Consider the following process model $N_2$: 
		
		\adjustbox{scale = 0.8}{
			\begin{tikzcd}	
				& {} && {} && {} \\
				\eye & {\fbox{a}} & \bigcirc & {\fbox{c}} & \bigcirc & {\fbox{g}} & \bigcirc \\
				\eye & {\fbox{b}} & \bigcirc & {\fbox{d}} && {} \\
				&& \bigcirc & {\fbox{e}} & \bigcirc & {\fbox{f}} & \bigcirc
				\arrow[from=2-1, to=2-2]
				\arrow[from=2-2, to=2-3]
				\arrow[from=2-3, to=2-4]
				\arrow[from=2-4, to=2-5]
				\arrow[from=3-1, to=3-2]
				\arrow[from=3-2, to=3-3]
				\arrow[from=3-3, to=3-4]
				\arrow[curve={height=6pt}, from=3-4, to=2-5]
				\arrow[from=2-5, to=2-6]
				\arrow[curve={height=6pt}, from=3-2, to=4-3]
				\arrow[from=4-3, to=4-4]
				\arrow[from=4-4, to=4-5]
				\arrow[from=4-5, to=4-6]
				\arrow[from=4-6, to=4-7]
				\arrow[from=2-6, to=2-7]
				\arrow["{[0, 2]}"{description, pos=0.1}, "\lrcorner"{text=white, anchor=center, pos=0.125, rotate=135}, draw=none, from=2-2, to=1-2]
				\arrow["{[3,3]}"{description, pos=0.1}, "\lrcorner"{text=white, anchor=center, pos=0.125, rotate=135}, draw=none, from=2-4, to=1-4]
				\arrow["{[0,1]}"{description, pos=0.1}, "\lrcorner"{text=white, anchor=center, pos=0.125, rotate=135}, draw=none, from=3-2, to=2-2]
				\arrow["{[6,8]}"{description, pos=0.1}, "\lrcorner"{text=white, anchor=center, pos=0.125, rotate=135}, draw=none, from=3-4, to=2-4]
				\arrow["{[6,8]}"{description, pos=0.1}, "\lrcorner"{text=white, anchor=center, pos=0.125, rotate=135}, draw=none, from=4-4, to=3-4]
				\arrow["{[0,1]}"{description, pos=0.1}, "\lrcorner"{text=white, anchor=center, pos=0.125, rotate=135}, draw=none, from=4-6, to=3-6]
				\arrow["{[3,5]}"{description, pos=0.1}, "\lrcorner"{text=white, anchor=center, pos=0.125, rotate=135}, draw=none, from=2-6, to=1-6]
			\end{tikzcd}
		}
		
		For this net and the word $$w = (a, 1)(b, 2)(c, 5)(d, 6)(e, 6)(f, 8)(g,8)$$ we can confuse transition labels with event labels, so we can define its flow function as $$a \xmapsto{}{} 1, b \xmapsto{}{} 2, c \xmapsto{}{} 4, d \xmapsto{}{} 4, e \xmapsto{}{} 4, f \xmapsto{}{} 2, g \xmapsto{}{} 2$$
		
		Where the order of events is clear (such as here, where it is alphabetical, or in linear causal processes)  we can write the flow function as a tuple, as follows : $$f_w = (1, 2, 4, 4, 2, 2)$$
	\end{example}
	
	Note that $e'$ as defined here is the latest causal predecessor of $e$
	. 
	Hence, $f_\tau$ as defined produces exactly the time durations that the guards of each transition in the model checks, i.e. the clock function values during the run. 
	
	As defined, we see that if a word is in the language of the model, then its $f$ function maps events to values that lie within the constraint that the event's corresponding transition demands, that is, it is perfectly aligned with the model with distance 0. 
	This condition is unfortunately not sufficient (due to urgency), but it is quite close to the exact condition necessary for a word to be in the language of a time Petri net. 
	
	Also note that given the underlying causal process and the resulting $f_\tau$, we can reconstruct $\tau$ quite straightforwardly as $$\forall e \in E : \tau(e) = \sum_{e' \leq_G e} f_\tau(e')$$
	
	\begin{lemma}[Delay only distance : $d_\theta$]\label{manhattan2}
		
		$d_\theta$ as defined previously is equivalent to the following formulation directly using the flow function:
		$$d_\theta(\tau_1, \tau_2) = \sum_{e \in E}|f_{\tau_1}(e) - f_{\tau_2}(e)|$$
	\end{lemma}	
	
	The flow function is hence a dual representation of timing functions, much like line graphs for graphs, as if a timing function traditionally labels its transition nodes with timestamps, the flow function labels edges leading up to transitions with the duration since said transition was enabled. This gives us a new way to formulate the alignment problem in terms of minimising distance from the flow vector of $\sigma$. 
	
	\begin{example}
		Let us study the process model $N_2$ from example \ref{wordandflow} once again : 
		Say we want to align  $$w = (a, 1)(b, 2)(c, 5)(d, 6)(e, 6)(f, 8)(g,8)$$ to $N_1$. We recall that we calculated its flow function to be $$a \xmapsto{}{} 1, b \xmapsto{}{} 2, c \xmapsto{}{} 4, d \xmapsto{}{} 4, e \xmapsto{}{} 4, f \xmapsto{}{} 2, g \xmapsto{}{} 2$$
		As membership in the language only depends on the flow values, we try to pick flow components that are as close to $w$'s as possible while keeping the alignment in the language, i.e. $\underset{u \in \mathcal{L}(N)}{minimize} {\sum_{e \in E} |f_w(e) - f_u(e)|}$. Choosing naively to optimise each summand gives us the flow function $$a \xmapsto{}{} 1, b \xmapsto{}{} 1, c \xmapsto{}{} 3, d \xmapsto{}{} 6, e \xmapsto{}{} 8, f \xmapsto{}{} 1, g \xmapsto{}{} 3$$ From which we can reconstruct the closest word $$u = (a, 1)(b, 1)(c, 4)(d, 7)(e, 7)(f, 8)(g, 10)$$
		
		Clearly if one tries to move any component of $u$'s flow function any closer to that of $w$, the word will no longer be in the language. 
		
	\end{example}
	We could reason this way because $d_\theta$ allows local edits that do not affect global membership in the language of the time Petri net, as both the time Petri net and delay edits concern themselves only with the duration of time elapsed between a transition's enabling and its firing. 
	This relationship is especially clear in extended free choice time Petri nets, where the only events that determine the fireability of an event are its causal history and events with the exact same preset which would naturally disable it by firing. 
	Hence, we present Algorithm \ref{DelayAlgo}, which minimizes the delay cost at each transition locally, while still overall ensuring membership in the language of the model. The main step in the algorithm is line 12, the rest is just to ensure the time process we obtain is a valid one.

	\begin{theorem}\label{DelayAlgoCorr} Algorithm \ref{DelayAlgo} is correct on extended free choice time Petri nets, that is, its result $\gamma$ has the properties : 
		\begin{enumerate}
			\item $(CN, p, \gamma)$ is a valid time process of $N$. 
			\item $d_\theta(\gamma, \sigma)  \leq  {d_\theta(x, \sigma)}$ for all $x$ such that $(CN, p, x)$ is a valid time process of $N$. 
		\end{enumerate}
	\end{theorem}

	\begin{remark} Note that Algorithm \ref{DelayAlgo} runs through the events of the causal process $E$ exactly once, and each time runs through the set of currently enabled transitions, and the total set of transitions, and hence its time complexity is linear in the size of the input $|CN|$, and the size of the transition set, i.e, $O(|E||T|)$. 
	\end{remark}
	\begin{figure}[!t]
		\removelatexerror
		\begin{algorithm}[H]
			\caption{Local $d_\theta$ Algorithm}
			\label{DelayAlgo}
			\begin{algorithmic}[1]
				\State \textbf{Input : }  $N$, $CN = (B, E, G), p$, $\sigma : E \to \mathbb{R}^+$
				\State \textbf{Output : }  $\gamma : E \to \mathbb{R}^+$ such that $\gamma$ is a valid timing function and $d_\theta(\gamma, \sigma)  \leq  {d_\theta(x, \sigma)}$ for all valid timing functions $x$ 
				
				\State $Curr = Min(CN)$
				\State $Enabled = \{t | \pre t \subseteq p(Curr)\}$
				\State $FD = \{(t, Eft(t), l, 0) | {t \in Enabled}, {l = \min_{\pre t' = \pre t}{Lft(t')}}\}$
				\State $Soon = \{{(t, eft, l, toe) \in FD} | {l+toe} \textrm{ is minimal in } FD\}$
				\While{ $E \neq \emptyset$} 
				\State Pick $(t, eft, l, toe) \in Soon, {t \in p(E)}$ 
				\State $E \gets E \setminus\{e | p(e) = t\}$
				\State $e \gets p^{-1}(t)$
				\State $Curr \gets \{Curr \setminus \{\pre e\}\} \cup \{e \pre\} $
				\State $f_{\gamma}(e) = \argmin_{x \in [eft, l]}|x - f_{\sigma}(e)|$
				\ForAll{$t' \in T \wedge \pre t' = \pre t$} 
				\State $Enabled \gets Enabled \setminus\{t'\}$
				\State $FD \gets FD \setminus\{(t', e', l', toe')\}$
				\EndFor
				\ForAll{$t' \in T \wedge t' \pre = t \pre$}
				\State $Enabled \gets Enabled \cup \{t'\}$
				\State $eft' \gets Eft(t')$
				\State $l' \gets \min_{\pre t'' = \pre t'}{Lft(t'')}$
				\State $toe' \gets toe + f_\gamma(e)$
				\State $FD \gets FD \cup \{(t', eft', l', toe')\}$
				\EndFor
				\State $Soon = \{{(t, eft, l, toe) \in FD} | {l+toe} \textrm{ is minimal in } FD \}$
				\EndWhile
				\State \textbf{return $\gamma$}
			\end{algorithmic}
		\end{algorithm}
		
	\end{figure}
	\section{The General Timed Alignment Problem} \label{align}
	
	Now that we have provided methods to tackle the purely timed alignment problem for two metrics, we return to the general timed alignment problem. 
	One approach for solving the general timed alignment problem would be to begin by completely ignoring the timed aspect, and aligning only the untimed part of the word $w$ and the process model $N$ (now reduced to the underlying untimed Petri net) using the $\mathbb{A}^*$ algorithm\cite{Adriansyah2014AligningOA}. This will produce 
	a word $w'$ with a valid causal process. This new untimed word $w'$ can then be given a timing function based on that of $w$, 
	repeating or deleting timestamps for any insertion or deletion moves. This yields a valid causal process and a potentially invalid timing function over it, which is then cast as a purely timed alignment problem. 
	
	This approach is quite naive, as it may have a tendency to exaggerate the importance of action label errors. For example, consider the following scenario : 
	
	\begin{example}
		Consider the following net $N_3$ where the final marking is presumed to be the sink place: 
		\adjustbox{scale=0.8}{
			\begin{tikzcd}
				& {} && {} && {} \\
				\eye_I & {\fbox{a}} & \bigcirc & {\fbox{a}} & \bigcirc & {\fbox{a}} &\bigcirc_f \\
				& {\fbox{b}} & \bigcirc & {\fbox{a}} & \bigcirc & {\fbox{a}} \\
				\arrow[from=2-1, to=2-2]
				\arrow[curve={height=12pt}, from=2-1, to=3-2]
				\arrow[from=2-2, to=2-3]
				\arrow[from=2-3, to=2-4]
				\arrow[from=2-4, to=2-5]
				\arrow[from=2-5, to=2-6]
				\arrow[from=2-6, to=2-7]
				\arrow[from=3-2, to=3-3]
				\arrow[from=3-3, to=3-4]
				\arrow[from=3-4, to=3-5]
				\arrow[from=3-5, to=3-6]
				\arrow[curve={height=12pt}, from=3-6, to=2-7]
				\arrow["{[0,0]}"{description, pos=0.2}, draw=white, from=2-2, to=1-2]
				\arrow["{[0,0]}"{description, pos=0.1}, draw=white, from=2-4, to=1-4]
				\arrow["{[0,0]}"{description, pos=0.1}, draw=white, from=2-6, to=1-6]
				\arrow["{[100,100]}"{description, pos=0.2}, draw=white, from=3-2, to=2-2]
				\arrow["{[100,100]}"{description, pos=0.2}, draw=white, from=3-4, to=2-4]
				\arrow["{[100,100]}"{description, pos=0.2}, draw=white, from=3-6, to=2-6]
			\end{tikzcd}
		}
		Now, let the word we seek to align be $w = (a, 100)(a, 100)(a, 100)$. 
		
		The alignment the approach we described above would provide is $w' = (a, 0)(a, 0)(a, 0)$. 
		
		But if we assign even a tenth of the cost of action edits to timestamp edits, clearly the closer word in the language was $w'' = (b, 100)(a, 100)(a, 100)$. 
	\end{example}
	
	The issue here is the tradeoff between timestamp edits and action label edits. A better approach would be to assign a particular cost $c_A$ to action edits and another $c_T$ to time edits, and devise an algorithm to minimise their sum. A potential fix to our previous idea would be to essentially retain the above approach, but instead of just doing it for the best untimed alignment, take a selection of good untimed alignments (up to some allowable threshold), and try to align timestamps for each untimed candidate, and then minimise the total cost over this set. This alleviates some of the bias towards preserving action labels by allowing for action deviations in case the timestamp alignment proves particularly easy, but this is still just an initial approach to the wider problem, and we hope to find better methods in the future. 
	
	\section{Implementation}
	
	We have implemented the stamp only alignment algorithm in python, available at \url{https://github.com/NehaRino/TimedAlignments}. This algorithm is quadratic in time complexity, and runs well in practice, as seen below 
	\centerline{
		\begin{tabular}{|c|c|c|c|}
			\hline
			Trace Length & 10 & 100 & 1000\\
			\hline
			Running Time (seconds) & 0.001 & 0.1 & 9.8 \\ 
			\hline
		\end{tabular}
	}
	
	
	\section{Perspectives and Conclusion}
	
	In this paper, we posed the alignment problem for timed processes, devised three metrics with which to study alignments for timestamp sequences, and solved the purely timed alignment problem for the first two metrics proposed. As far as we know, this is the first step in conformance checking for time-aware process mining, and much further work can be inspired from this point. The alignment problem for the third metric $d_N$ is a first future direction. Secondly, for both metrics studied here the class of models for which the alignment problem was solved efficiently are structurally restricted (being linear causal processes for $d_t$ and extended free choice time Petri nets for $d_\theta$) and it would be interesting to see how to broaden the scope of these methods to larger classes of process models. Thirdly, further investigation in the general timed alignment problem is necessary, as our proposed approach here is rather rudimentary and can certainly be improved. Lastly, there are a number of other conformance artefacts that can be set and studied in the timed setting, such as anti-alignments \cite{CBC21}, and one can better develop all such conformance checking methods to account for timed process models.
	
	\bibliographystyle{IEEEtran}
	\bibliography{Conformanceieee}
	\clearpage
	\appendix 
	\subsection{Stamp Only Algorithm}
	We start off with a proof of Lemma \ref{manhattan}
	
	\begin{proof}
		The way we see this is as follows. 
		Given a causal process $CN = (B, E, G)$ and two timing functions (not necessarily constituting valid timings) $\tau_1, \tau_2$, let $$d_t(\tau_1, \tau_2) = \sum_{e \in E} |\tau_1(e) - \tau_2(e)|$$
		This is both achievable by stamp only moves by doing the corresponding $\tau_1(e) - \tau_2(e)$ stamp move to the second word, and is the least cost any such stamp only path can take as Manhattan distance obeys the triangle inequality, and every stamp move results in a word at exactly the same Manhattan distance from its predecessor as the cost of said stamp move. 
		
	\end{proof}
	
	Now, we can proceed to the proof of Theorem \ref{convexset} : 
	
	\begin{proof}
		Lemma 1 (General Form Lemma) of \cite{BM82} states that the firing domains of state classes for any bounded time Petri net may be expressed as solution sets of systems of inequalities of the following form :  $$\begin{cases}
			a_i \leq t(i) \leq b_i & \forall i \\ 
			t(j) - t(k) \leq c_{jk} & \forall j \neq k
		\end{cases}$$
		
		As $\gamma \in \mathcal{L}(N)$  this means it obeys the set of inequalities corresponding to the state class containing the final marking. Now, the above constraint is convex, so $\gamma$ varies over a convex set. As $\sigma$ is a constant vector, translation by it will keep the space convex, so $\gamma - \sigma \in \mathcal{C}$ for some convex set $\mathcal{C}$. 
		
		Now, we wish to minimize $|\gamma - \sigma|$, but it would be better to view this objective as a linear function. To this end we use a trick from \cite{BV2014} and use the following variable translation : $$\alpha, \beta \geq 0 : \alpha - \beta = \gamma - \sigma$$
		
		As defined above, $\alpha$ and $\beta$ are the absolute values of the nonnegative and negative components of the vector $\gamma - \sigma$, respectively, so their difference is $\gamma - \sigma$ and their sum is $|\gamma - \sigma|$. Clearly now, we have completed the transformation, and our alignment problem is equivalently cast as $$\underset{\alpha - \beta \in \mathcal{C}, \alpha, \beta \geq 0}{\mathsf{minimize}} \hspace{0.5em} \alpha + \beta  $$
		
	\end{proof}
	
	Before we prove Theorem \ref{StampAlgoCorr}, we prove a small, relevant lemma. 
	
	\begin{lemma}\label{minpwc}
		For any convex piecewise linear function $g$, the following function $g'$ is also convex and piecewise linear. 
		$$g'(d) = \min_{d-d' \in [a,b]} g(d')$$
		
	\end{lemma}
	
	\begin{proof}
		For any piecewise linear convex function $g$, the function $$g'(d) = \min_{d-d' \in [a,b]} g(d')$$ is also piecewise linear. 
		If the domain of $g$ was $[l, h]$, the domain of $g'$ is $[l + a, h + b]$. 
		
		Say the leftmost minimum of $g$ (if $g$ has multiple minima they are all on the same flat line segment) is at $m \in [l,h]$. 
		For all $x \in [l + a, m+ a] $ and $x - y \in [a, b]$ and $y \in [l,h]$ implies $y \in [x-b+a, x-a] \cap [l, h] \subseteq [l, h] \cap [l-b+a m] = [l,m]$ meaning for each $x$, in the possible domain of $y$ $g(y)$ is a strictly decreasing function, hence  we have $$\forall x \in [l + a, m+ a] : \min_{x - y \in [a,b]} g(y) = g(x - a) $$
		
		Now, let the rightmost minimum of $g$ is $m'$. 
		For all $x \in [m + a, m' + b] $ and $x - y \in [a, b]$ and $y \in [l,h]$ implies $y \in [x-b, x-a] \cap [l, h]$, and $[x - b, x - a] \cap [m, m'] \neq \emptyset$. 
		Now, this means for each $x \in [m+a, m' + b]$ the domain of $y$ has a minimum of $g$, so we have $$\forall x \in [m + a, m'+ b] : \min_{x - y \in [a,b]} g(y) = g(m) $$
		
		Lastly, let  $x \in [m' + b, h + b] $ and $x - y \in [a, b]$ and $y \in [l,h]$. 
		This implies $y \in [x-b, x-a] \cap [l, h] = [m', h]$, so for each $x \in [m+a, m' + b]$ the domain of $y$ is in a region where $g$ is strictly increasing, so we have $$\forall x \in [m + a, m'+ b] : \min_{x - y \in [a,b]} g(y) = g(x+b) $$
		
		So we see that for any convex piece-wise linear function $g$, $$\min_{d-d' \in [a,b]} g(d') = \begin{cases}
			g(d-a) & l + a \leq d < m + a \\
			g(m) & m+a \leq d \leq m' + b \\
			g(d+b) & m' + b < d \leq h + b   \\
		\end{cases}$$
		
		Hence $g'$ is also convex and piecewise linear, and has at most one more linear segment than $g$. 
	\end{proof}
	
	We include here pseudocode for the \textsc{backTrack} subfunction : 
	
	\begin{figure}[!t]
		\removelatexerror
		\begin{algorithm}[H]
			\caption{Backtrack to find the best word}
			\label{AuxFunc3}
			\begin{algorithmic}
				\State \textbf{Input : }  $\sigma$, $graphlist$, $cost$, $SI$
				\State \textbf{Output : }  $\gamma \in \mathcal{L}(N)$ such that $d_t(\gamma, \sigma) = cost$
				
				\For{$i = 0, i < |graphlist|$}
				\State Find $ylist[i]$, the $y$ values of $graphlist[i]$
				\EndFor
				
				\State $i \gets |sigma|-1$
				\While{$i \geq 0$}
				\State $[a, b] = SI[i+1]$ 
				\State Truncate $ylist[i]$ and $graphlist[i]$ so $ \gamma_{i+1} - x \in [a, b]$. 
				\State $k \gets $ \textsc{binarySearch}($ylist[i], cost$)
				\State $xsegment \gets graphlist[i][k]$
				
				\State $\gamma_i \gets xsegment^{-1}(cost)$
				
				\State $cost \gets cost - |\gamma_i - \sigma_i|$
				\State $i \gets i-1$
				\EndWhile
				
			\end{algorithmic}
		\end{algorithm}
	\end{figure}
	
	We may now prove Theorem \ref{StampAlgoCorr} as follows : 
	
	\begin{proof}
		We are given a process model $N$, an observed trace $\sigma$, and its underlying linear causal process $CN = (B, E, G)$ with homomorphism $p$ for the untimed run of $\sigma$ on $N$.
		As discussed earlier, for linear process models, we define $$g_i(d) = \min_{\tau|_i \in \mathcal{L}^i(N) \wedge \tau_i = d} cost(\tau|_i)$$
		
		This family of functions is seen to obey the following recursion: 
		
		$$g_{i+1}(d) = |d - \sigma_{i+1}| + \min_{d-d' \in [a,b]} g_i(d')$$
		
		If we can calculate the family of functions $\forall i \leq n : g_i$ efficiently, then the minimal cost for aligning $\sigma$ to $N$ is simply $$\min_{d} g_n(d)$$
		
		Hence to prove the correctness of Theorem \ref{StampAlgoCorr} we only need show that the auxiliary subfunctions called do indeed do their job correctly, as the main function just implements the recursion. 
		
		To this end, we first notice that the family of functions $g_i(d)$ are convex and piecewise linear. 
		This can be seen through induction : 
		
		Firstly $g_1(d) = |d - \sigma_1|$ so the base case is covered. 
		
		Now suppose $g_{i-1}$ is known to be convex and piecewise linear for some $i \geq 2$. 
		By Lemma \ref{minpwc} we see that $\min_{d-d' \in [a_i,b_i]} g_{i-1}(d')$  is also convex and piecewise linear, and so is $|d - \sigma_i|$ and both these properties are preserved by addition, so $g_i(d)$ must also be convex and piecewise linear. 
		
		Now that we have this result, our representation of the graph as a sequence of line segments, each represented as a triple (left x endpoint, slope, right x endpoint) is justified. 
		The y value of any graph $g_i$ can be back calculated because we know the value of the cost of the leftmost endpoint, it represents the word obtained by firing every transition at the earliest possible time. 
		
		Now subfunction \textsc{graphAddMod} is adding the point $\sigma_i$ if it is inside the domain of the function, and changing the slopes of the segments before $\sigma_i$ by subtracting one, and of those after $\sigma_i$ by adding one, which has the effect of adding the modulus component $|d - \sigma_i|$ to the previous graph. 
		
		Subfunction \textsc{graphMin} is even simpler, as Lemma \ref{minpwc} suggests all it does is translate the whole graph forward by $a$, and then translate the strictly increasing portion from $[m', h]$ by an addition $b-a$, adding an extra flat segment of length $b-a$ in between to keep the two portions of the function connected, and so on input $g(d), [a, b]$ it outputs $\underset{d-d' \in [a,b]}{\min}g(d')$. 
		
		Subfunction \textsc{backTrack} takes the minimum cost obtained by the calculation of the $g_n$ function, finds the value of $d$, i.e. $\gamma_n$ for which the minimum is achieved, and subtracts the cost aligning only the last place incurs, thereby finding the minimum cost for aligning the $n-1$ length prefix. 
		It then proceeds to do the same iteratively for each prefix, reverse-engineering a trace $\gamma$ for which the total minimum cost of aligning to $\sigma$ is achieved. 
		
	\end{proof}
	
	\subsection{Preliminaries for Delay Only Algorithm}

	Now, in order to study \textit{timed} executions, we want to be able add a timing function to the causal process we defined above, thereby allowing us to record when different transitions are taken. 
	The following definitions, properties and theorems in this subsection cover the material developed in Aura and Lilius' article on Time Processes \cite{AL97}. 
	
	Given the definition of the timing function, we try to see how a time process unfolds. 
	First the initial events, that is, those enabled at $Min(CN)$ happen, and as each event occurs, new events are enabled and disabled. 
	In order to keep track of this, we first define the auxiliary Cut function that represents the effect of having all the events in a subset fire simultaneously (we only use Cut on sets of events that can concurrently fire), as defined below. 
	$$Cut(E') = (E' \pre  \cup Min(NS)) \setminus \pre  E'.$$
	
	Now we can define the  \textit{time of enabling} for a transition $t$ of a time Petri net in a set of conditions $B'$ of its causal process as $TOE(B', t) = $ $$ max(\{\tau(\pre b) | b \in B' \setminus Min(CN) \wedge p(b) \in \pre  t\} \cup \{0\})$$
	
	We must now verify that this timing function does indeed represent a valid execution that obeys all the static interval constraints, as below : 
	
	\begin{definition}[Valid Timing] \label{validt} A timing function $\tau$ is a \textit{valid timing} of the causal process iff 
		$$\forall e \in E : \tau(e) \geq TOE(\pre e, p(e) ) + Eft(p(e))$$
		$$\forall e \in E : \forall t \in Enabled (p(C_e)) : \tau(e) \leq TOE(C_e, t) + Lft(t)$$ 
		where $C_e = Cut(Earlier(e))$. 
	\end{definition}
	
	Here, checking every element of $C_e$ might seem unnecessary but in certain nets, complex dependencies between transitions can cause causally unrelated transitions to force transitions to fire or be disabled, due to urgency. 
	This phenomenon is known as confusion, and hence to guard against this, a timing function must check that it never leaves a section of the process behind, completing the firing or disabling of every event in $Earlier(e)$ before it can reason about the fireability of $e$. 
	
	\begin{definition}[Time Process] \label{tp} A \textit{time process} of a time Petri net $N$ is a triple $(CN, p, \tau)$ where $\tau$ is a valid timing of $(CN, p)$ which is a causal process of $N$. 
	\end{definition}
	
	A few last notions here will assist us when we consider a simpler class of time Petri nets, that allow for much easier validity checking. 
	
	\begin{definition}[Extended Free Choice]\label{efc} A time Petri net is \textit{extended free choice} iff for all two transitions $t$ and $t'$, $\pre t \cap \pre t' \neq \emptyset$ implies $\pre t = \pre t'.$
	\end{definition}
	
	This class of time Petri nets ensure that the net is \textit{confusion-free}, that is, causally unrelated events cannot affect the fireability of other events. 
	This means that in order to check for the validity of a timing function, checking all of $C_e$ as before is no longer necessary. 
	
	How do we construct a valid timing function on the fly? 
	Given a partial time process of a TPN, we want to be able to study the effect of firing particular transitions amongst the set of transitions currently enabled. 
	In order to do so, we notice the following class of transitions. 
	
	A transition $t$ is a \textit{choice} at $B' \subseteq B$ iff $B'$ is a co-set that maps injectively to places and $\pre t = p(B')$, and a choice is an \textit{extension} of the process iff $B' \subseteq Cut(E)$. 
	
	These transitions reflect exactly the transitions enabled at a particular moment in the evolution of the process. Either they must be fired or disabled, by the urgency condition. 
	
	Given the above, Aura and Lilius characterise a method by which one can build partial time processes inductively, building forward while maintaining the validity of the timing function, denoted as keeping the process complete with respect to the timing function : 
	
	\begin{definition} A causal process $(CN, p)$ of a time Petri Net where $CN = (B, E, G)$ is said to be \textit{complete} with respect to a timing function $\tau$ iff for every extension transition $t$ of the process, $$\max\{\tau(e)| e \in E\} \leq TOE(Cut(E), t) + Lft(t).$$
	\end{definition}
	
	Now given this, the following theorem due to Aura and Lilius \cite{AL97} holds : 
	
	\begin{theorem}\label{thm35}
		Let $N$ be an extended free choice time Petri net and $(CN, p)$ a causal process of $N$, where $CN = (B, E, G)$. A timing function $\tau$ is valid iff the following criteria hold : 
		
		\begin{enumerate}
			\item$ \forall e \in E :  Eft(p(e)) \leq \tau(e) - TOE(\pre e, p(e))$
			\item $\tau(e) - TOE(\pre e, p(e)) \leq  \min \{Lft(t) | \pre t = \pre p(e)\}$
			\item $(CN, p)$ is complete with respect to $\tau$, i.e, for every extension transition $t$ of the process, $$\max\{\tau(e)| e \in E\} \leq TOE(Cut(E), t) + Lft(t)$$
			
		\end{enumerate}
	\end{theorem}
	
	\subsection{Proof of correctness of Delay Only Algorithm}
	We start off with a quick proof for Lemma \ref{manhattan2} : 
	
	\begin{proof}
		By noting that delay moves on a trace $\tau$ translate exactly to stamp moves on $f_\tau$, we can deduce by the same argument as Lemma \ref{manhattan} that the distance function $d_\theta$ can be seen to be identical to the Manhattan distance on flow functions, i.e, the above formulation. 
	\end{proof}

	\begin{proof}
		By theorem \ref{thm35} proved in Aura and Lilius' article \cite{AL97}, we know that in order to ensure that this is a valid time process, we need only ensure three inequalities hold as the time process $\gamma$ evolves. 
		
		\begin{enumerate}
			\item$ \forall e \in E :  Eft(p(e)) \leq \gamma(e) - TOE(\pre e, p(e))$
			
			By definition $f_\gamma(e) = \gamma(e) - TOE(\pre e, p(e))$, and by line 6 of the algorithm the variable $eft$ is correctly initialised to store $Eft(p(e))$, and line 17 ensures the above inequality holds. 
			
			\item $\gamma(e) - TOE(\pre e, p(e)) \leq  \min \{Lft(t) | \pre t = \pre p(e)\}$
			
			In similar vein, line 6 also ensures that $l$ is initialised precisely to store $\min \{Lft(t) | \pre t = \pre p(e)\}$, and line 17 again enforces the left hand side to be less than $l$. 
			
			\item $(CN, p)$ is complete with respect to $\gamma$, i.e, for every extension transition $t$ of the process, $$\max\{\gamma(e)| e \in E\} \leq TOE(Cut(E), t) + Lft(t)$$
			
			At every point when assigning a timestamp to an event, it is ensured that it belongs to $Soon$, a set defined to contain only those enabled events that have the least value of $l + toe$, or, $TOE(Cut(E), p(e)) + Lft(p(e))$. 
			This ensures that whenever an event $e$ is assigned a timestamp $\gamma(e)$, for all extension transitions $t$  (which must either be causal descendents of $e$ or already enabled when $e$ was picked, thereby being in $Enabled$, by extended free choice) we know that $\gamma(e) \leq TOE(Cut(E), p(e)) + Lft(p(e)) \leq TOE(Cut(E), t) + Lft(t)$. 
			
			Hence, $(CN, p, \gamma)$ is a valid time process of $N$. 
			
		\end{enumerate}
		
		As for its optimality, we see that any other valid time process would have to obey the same inequalities, in particular for another valid timing function $\gamma'$ it would have to obey $$f_{\gamma'}(e) \in [Eft(p(e)), \min_{\pre t' = \pre p(e)}{Lft(t')}] = [a,b]$$ and hence at any event $e$ where $\gamma$ differs from $\gamma'$ it would have lower than or equal cost for that event as $$f_\gamma(e) = \argmin_{x \in [a, b]} |f_\sigma(e) - x|$$
		
		Hence, the algorithm is correct. 
	\end{proof}
	
	\subsection{A Note on Linear Causal Processes}
	We take a moment to discuss the structurally restricted class of processes that our stamp only-algorithm works for, that is, causal processes whose graphical structure  is that of a straight line. 
	These reflect the executions of time Petri nets that do not have any branching points ($e \in E : |e \pre | > 1$) and naturally as a consequence, no points of synchrony ($e \in E : |\pre e| > 1$). 
	This essentially reflects the executions of automata, allowing for exclusive branching ($p \in P : |p \pre| > 1$) and iteration (that is, cycles in the model) only. 
	It loses the ability to capture properties of a concurrent nature, as the run never splits into more than one token. 
	On the other hand, studying the language of this restricted class and the alignment problem over it becomes substantially simpler, the $G$ pre-order becomes total, and so to begin with, this is a significantly more tractable class of time processes. 
	
	\begin{example}\label{Branchingcp}
		Going back to the net $N$ in example \ref{TPN}, and the firing sequence we analysed then $$w = (a, 1)(b, 2)(d, 3)(e, 4)(f,5)$$ 
		
		We now construct the causal process of this execution, and see that it is indeed branching, as the very first transition, $a$, itself has two post-places, which violates the no branching condition $e \in E : |e \pre | \leq  1$. 
		\adjustbox{scale = 0.8}{
			\begin{tikzcd}
				\bigcirc & {\fbox{a}} & \bigcirc & {\fbox{b}} & \bigcirc & {\fbox{f}} & \bigcirc \\
				&& {\fbox{d}} & \bigcirc & {\fbox{e}} & \bigcirc
				\arrow[from=1-1, to=1-2]
				\arrow[from=1-2, to=1-3]
				\arrow[from=1-3, to=2-3]
				\arrow[from=2-3, to=2-4]
				\arrow[from=1-3, to=1-4]
				\arrow[from=1-4, to=1-5]
				\arrow[from=1-5, to=1-6]
				\arrow[from=2-4, to=2-5]
				\arrow[from=2-5, to=2-6]
				\arrow[from=2-6, to=1-6]
				\arrow[from=1-6, to=1-7]
		\end{tikzcd}}
	\end{example}
	
	\begin{example}\label{Linearcp}
		On the other hand, consider the following time Petri net $N_1$ : \\
		\adjustbox{scale = 0.8}{\begin{tikzcd}
				{} & {\fbox{a}} & {\fbox{e}} & \bigcirc & {\fbox{a}} & \bigcirc & {\fbox{a}} & \bigcirc \\
				& \bigcirc & {\fbox{c}} & \bigcirc & \bigcirc & {\fbox{c}} & \bigcirc & {\fbox{b}} \\
				{} & {\fbox{b}} & \bigcirc & {\fbox{d}} & {\fbox{e}} & \bigcirc & {\fbox{a}} & \bigcirc
				\arrow[from=2-2, to=1-2]
				\arrow[curve={height=12pt}, from=1-2, to=2-2]
				\arrow[from=2-2, to=3-2]
				\arrow[from=3-2, to=3-3]
				\arrow[from=3-3, to=2-3]
				\arrow[from=3-3, to=3-4]
				\arrow[from=3-4, to=2-4]
				\arrow[from=2-4, to=1-3]
				\arrow[from=1-3, to=2-2]
				\arrow[from=2-3, to=2-4]
				\arrow["{[0,1]}"{description, pos=0.3}, draw=white, from=1-2, to=1-1]
				\arrow["{[2,2]}"{description, pos=0.3}, draw=white, from=3-2, to=3-1]
				\arrow["{[2,3]}"{description, pos=0.2}, draw=white, from=2-3, to=1-3]
				\arrow[from=1-4, to=1-5]
				\arrow[from=1-5, to=1-6]
				\arrow[from=1-6, to=1-7]
				\arrow[from=1-7, to=1-8]
				\arrow[from=1-8, to=2-8]
				\arrow[from=2-8, to=2-7]
				\arrow[from=2-7, to=2-6]
				\arrow[from=2-6, to=2-5]
				\arrow[from=2-5, to=3-5]
				\arrow[from=3-5, to=3-6]
				\arrow[from=3-6, to=3-7]
				\arrow[from=3-7, to=3-8]
				\arrow["{[2,2]}"{description, pos=0.3}, draw=white, from=1-3, to=1-4]
				\arrow["{[0,0]}"{description, pos=0.3}, draw=white, from=3-4, to=3-5]
		\end{tikzcd}}

		Now, as all the branching and joining happens at places rather than transitions, its executions all have linear causal processes, such as the above causal process for the untimed word $w = aabcea$. 
	
		This causal process is unrolled from the Petri net, exactly the way runs of finite state automata are simple paths over the graph of the automaton. We throughout assume that wherever needed, such a causal process can be obtained efficiently. 
	\end{example}
	
\end{document}